\documentclass[11pt]{article}
\usepackage{times}
\usepackage{amssymb}
\usepackage{amsmath}
\usepackage{xspace,multirow}
\usepackage[shortlabels]{enumitem}
\usepackage{calc}
\usepackage{amsthm}

\newenvironment{proofof}[1]{\begin{proof}[Proof of {#1}]}{\end{proof}}
\newenvironment{proofsketch}{\begin{proof}[Proof sketch]}{\end{proof}}

\newtheorem{theorem}{Theorem}[section]
\newtheorem{lemma}[theorem]{Lemma}

\newcommand{\np}{{\em NP}\xspace}
 
\newcommand{\npcomplete}{\np-complete\xspace}

\DeclareMathOperator{\supp}{supp}

\DeclareMathOperator{\argmax}{argmax}

\newcommand{\R}{\ensuremath{\mathbb R}}
\newcommand{\Z}{\ensuremath{\mathbb Z}}

\newcommand{\A}{\ensuremath{\mathcal{A}}}
\newcommand{\B}{\ensuremath{\mathcal{B}}}

\newcommand{\G}{\ensuremath{\mathcal{G}}}
\newcommand{\Hc}{\ensuremath{\mathcal{H}}}

\newcommand{\Lc}{\ensuremath{\mathcal L}}

\newcommand{\Sc}{\ensuremath{\mathcal S}}
\newcommand{\Pc}{\ensuremath{\mathcal P}}

\newcommand{\OPT}{\ensuremath{\mathit{OPT}}}

\newcommand{\sm}{\ensuremath{\setminus}}
\newcommand{\es}{\ensuremath{\emptyset}}

\newcommand{\poly}{\operatorname{\mathsf{poly}}}

\newcommand{\e}{\ensuremath{\epsilon}}

\newcommand{\sse}{\subseteq}

\newcommand{\tx}{\ensuremath{\tilde x}}

\newcommand{\hx}{\ensuremath{\hat x}}

\newcommand{\hy}{\ensuremath{\hat y}}

\newcommand{\hmu}{\ensuremath{\hat\mu}}

\newcommand{\ld}{\ensuremath{\lambda}}

\newcommand{\al}{\ensuremath{\alpha}}

\newcommand{\dt}{\ensuremath{\delta}}

\newcommand{\w}{\ensuremath{\omega}}

\newcommand{\ozalg}{\ensuremath{\A_{\text{OZ}}}}
\newcommand{\bnalg}{\ensuremath{\B_{\text{BN}}}}
\newcommand{\rndalg}{\ensuremath{\B}}
\newcommand{\absopt}{\ensuremath{\mathsf{opt}}}
\newcommand{\ldobj}{\G}
\newcommand{\absldobj}{\Hc}
\newcommand{\cst}{CCST\xspace}
\newcommand{\mcst}{M\cst}
\newcommand{\mbdst}{MBDST\xspace}
\newcommand{\fpra}{FPRA\xspace}
\newcommand{\pst}{\ensuremath{\text{P}_{\text{ST}}}}

\newcommand{\PM}{\ensuremath{\mathcal{P}(M)}}
\newcommand{\QPM}{\ensuremath{(Q^{\PM}})}

\title{Approximating Min-Cost Chain-Constrained Spanning Trees: A Reduction from Weighted
  to Unweighted Problems\footnote{A preliminary version~\cite{LinharesS16} will appear in
    the Proceedings of the 18th IPCO, 2016.}}
\author{Andr\'e Linhares\thanks{{\tt \{alinhare,cswamy\}@uwaterloo.ca}.
        Dept. of Combinatorics and Optimization, University of Waterloo, Waterloo, ON N2L 3G1.
        Research supported partially by the second author's NSERC grant 327620-09, NSERC
        Discovery Accelerator Supplement Award, and Ontario Early Researcher Award.} 
\and 
        Chaitanya Swamy$^{\text{\thefootnote}}$
}

\date{}

\begin{document}

\maketitle

\begin{abstract}
We study the {\em min-cost chain-constrained spanning-tree} (abbreviated \mcst) problem:
find a min-cost spanning tree in a graph subject to degree constraints on a nested family
of node sets. 
We devise the {\em first} polytime algorithm that finds a spanning tree that (i) violates
the degree constraints by at most a constant factor {\em and} (ii) whose cost is
within a constant factor of the optimum. 
Previously, only an algorithm for {\em unweighted} \cst was known~\cite{olver}, which
satisfied (i) but did not yield any cost bounds. This also yields the first result that
obtains an $O(1)$-factor for {\em both} the cost approximation and violation of
degree constraints for any spanning-tree problem with general degree bounds on node sets, 
where an edge participates in a super-constant number of degree constraints.

A notable feature of our algorithm is that we {\em reduce} \mcst to unweighted \cst (and
then utilize~\cite{olver}) 
via a novel application of {\em Lagrangian duality} to simplify the {\em cost structure}
of the underlying problem and obtain a decomposition into certain uniform-cost
subproblems.  

We show that this Lagrangian-relaxation based idea is in fact applicable more generally
and, for any cost-minimization problem with packing side-constraints, yields a reduction
from the weighted to the unweighted problem. We believe that this reduction is of
independent interest. As another application of our technique, 
we consider the {\em $k$-budgeted matroid basis} problem, where we build upon a recent 
rounding algorithm of~\cite{BansalN16} to obtain an improved $n^{O(k^{1.5}/\e)}$-time 
algorithm that returns a solution that satisfies (any) one of the budget constraints
exactly and incurs a $(1+\e)$-violation of the other budget constraints. 
\end{abstract}

\section{Introduction}
Constrained spanning-tree problems, where one seeks a minimum-cost spanning tree
satisfying additional ($\{0,1\}$-coefficient) packing constraints, constitute an important 
and widely-studied class of problems. 
In particular, when the packing constraints correspond to node-degree bounds, we obtain
the classical {\em min-cost bounded-degree spanning tree} (\mbdst) problem, which
has a rich history of
study~\cite{FurerR94,KonemannR02,KonemannR03,ChaudhuriRRT09,Goemans06,SinghL07} 
culminating in the work of~\cite{SinghL07} that yielded an optimal result for \mbdst.
Such degree-constrained network-design problems arise in diverse areas including VLSI
design, vehicle routing and communication networks (see, e.g., the references
in~\cite{RaviMRRH01}), and their study has led to the development of powerful techniques
in approximation algorithms.   

Whereas the {\em iterative rounding and relaxation} technique introduced
in~\cite{SinghL07} (which extends the iterative-rounding framework of~\cite{Jain01})
yields a versatile technique for handling node-degree constraints 
(even for more-general network-design problems), 
we have a rather limited understanding of 
spanning-tree problems with more-general degree constraints,
such as constraints $|T\cap \dt(S)|\leq b_S$ for sets $S$ in some (structured) 
family $\Sc$ of node sets.%
\footnote{Such general degree constraints arise in the context of finding 
{\em thin trees}~\cite{AsadpourGMOS10}, 
where $\Sc$ consists of all node sets, which turn out to be a very useful tool in
devising approximation algorithms for {\em asymmetric TSP}.} 
A fundamental impediment here is our inability to leverage the techniques
in~\cite{Goemans06,SinghL07}. 
The few known results yield: (a) (sub-) optimal cost, but a {\em super-constant} additive-
or multiplicative- factor violation of the degree
bounds~\cite{BansalKN09,AsadpourGMOS10,ChekuriVZ10,BansalKKNP13};    
or 
(b) a multiplicative $O(1)$-factor violation of the degree bounds (when $\Sc$ is a nested
family), but {\em no cost guarantee}~\cite{olver}.  
In particular, in stark contrast to the results known for node-degree-bounded
network-design problems, there is no known algorithm that yields an $O(1)$-factor cost 
approximation {\em and} an (additive or multiplicative) $O(1)$-factor violation of the
degree bounds. (Such guarantees are only known 
when each edge participates in $O(1)$ degree constraints~\cite{BansalKKNP13}; see
however~\cite{Zenklusen12} for an exception.) 

We consider the {\em min-cost chain-constrained spanning-tree} (\mcst) problem introduced
by~\cite{olver}, 
which is perhaps the most-basic setting involving general degree bounds where there is a 
significant gap in our understanding vis-a-vis node-degree bounded problems. 
In \mcst, we are given an undirected connected graph $G = (V,E)$, nonnegative edge costs 
$\{c_e\}$, a nested family $\Sc$ (or {\em chain}) of node sets 
$S_1 \subsetneq S_2 \subsetneq \dots \subsetneq S_{\ell}\subsetneq V$, and integer degree 
bounds $\{b_S\}_{S\in\Sc}$. 
The goal is to find a minimum-cost spanning tree $T$ such that $|\dt_T(S)|\leq b_S$ for
all $S\in\Sc$, where $\dt_T(S):=T\cap\dt(S)$.
Olver and Zenklusen~\cite{olver} give an algorithm for {\em unweighted \cst} that returns
a tree $T$ such that $|\dt_T(S)|=O(b_S)$ (i.e., there is {\em no bound on $c(T)$}), and
show that, for some $\rho>0$, it is \npcomplete to obtain an additive
$\rho\cdot\frac{\log|V|}{\log\log|V|}$ violation of the degree bounds. 
We therefore focus on bicriteria $(\al,\beta)$-guarantees for \mcst, where the tree $T$
returned satisfies $c(T)\leq\al\cdot\OPT$ and $|\dt_T(S)|\leq\beta\cdot b_S$ for all
$S\in\Sc$. 

\paragraph{Our contributions.}
Our main result is the {\em first} $\bigl(O(1),O(1)\bigr)$-approximation
algorithm for \mcst. Given any $\ld>1$, our algorithm returns a tree $T$ with
$c(T)\leq\frac{\ld}{\ld-1}\cdot\OPT$ {\em and} $|\dt_T(S)|\leq 9\ld\cdot b_S$ for all
$S\in\Sc$, using the algorithm of~\cite{olver} for unweighted \cst, denoted $\ozalg$, as a
black box (Theorem~\ref{mcstthm}).  
As noted above, this is also the {\em first} algorithm that achieves an
$\bigl(O(1),O(1)\bigr)$-approximation 
for any spanning-tree problem with general degree constraints 
where an edge belongs to a super-constant number of degree constraints.

We show in Section~\ref{genredn} that our techniques 
are applicable more generally. 
We give a {\em reduction} showing that for {\em any} cost-minimization problem with
packing side-constraints, 
if we have an algorithm for the {\em unweighted} problem that returns a solution with
an $O(1)$-factor violation of the packing constraints and satisfies a certain property, 
then one can utilize it 
to obtain an $\bigl(O(1),O(1)\bigr)$-approximation
for the cost-minimization problem. 
Furthermore, we show that if the algorithm for the unweighted counterpart satisfies a
stronger property, then we can utilize it to obtain a 
$\bigl(1,O(1)\bigr)$-approximation (Theorem~\ref{twosidethm}). 

We believe that our reductions are of independent interest and will be useful in 
other settings as well. Demonstrating this, we show an application to the 
{\em $k$-budgeted matroid basis} problem, wherein we seek to find a basis satisfying $k$
budget constraints. Grandoni et al.~\cite{Grandoni2014} devised an $n^{O(k^2/\e)}$-time
algorithm that returned a $(1,1+\e,\ldots,1+\e)$-solution: i.e., the solution satisfies
(any) one budget constraint exactly and violates the other budget constraints by a
$(1+\e)$-factor (if the problem is feasible). Very recently, Bansal and
Nagarajan~\cite{BansalN16} improved the running time to $n^{O(k^{1.5}/\e)}$ but return only 
a $(1+\e,\ldots,1+\e)$-solution.  
Applying our reduction (to the algorithm in~\cite{BansalN16}), we obtain the 
{\em best of both worlds}: we return a  $(1,1+\e,\ldots,1+\e)$-solution in
$n^{O(k^{1.5}/\e)}$-time (Theorem~\ref{bmbthm}). 

\medskip
The chief novelty in our algorithm and analysis, and the key underlying idea, is an 
unorthodox use of {\em Lagrangian duality}. 
Whereas typically Lagrangian relaxation is used to drop complicating constraints and
thereby simplify the constraint structure of the underlying problem, in contrast, we use
Lagrangian duality to simplify the {\em cost structure} of the underlying problem by
equalizing edge costs in certain subproblems.  
To elaborate (see Section~\ref{overview}), the algorithm in~\cite{olver} for
unweighted \cst can be viewed as taking a solution $x$ to the natural
linear-programming (LP) relaxation for \mcst, converting it to another feasible solution
$x'$ satisfying a certain structural property, and exploiting this property to round $x'$
to a spanning tree. The main bottleneck here in handling costs (as also noted
in~\cite{olver}) is that 
$c^\intercal x'$ could be much larger than $c^\intercal x$ since the conversion 
ignores the $c_e$s and works with an alternate ``potential'' function.

Our crucial insight is that {\em we can exploit Lagrangian duality to obtain 
perturbed edge costs $\{c^{y^*}_e\}$ such that the change in perturbed cost due to the   
conversion process is bounded}.  
Loosely speaking, 
if 
the conversion process shifts weight from $x_f$ to $x_e$, then we ensure 
that $c^{y^*}_e=c^{y^*}_f$ (see Lemma~\ref{equalize}); thus,
$(c^{y^*})^\intercal x=(c^{y^*})^\intercal x'$!
The perturbation also ensures that applying $\ozalg$ to $x'$ yields a tree whose
perturbed cost is equal to $(c^{y^*})^\intercal x'=(c^{y^*})^\intercal x$. 
Finally, we show that for an optimal LP solution $x^*$, the ``error'' $(c^{y^*}-c)^\intercal x^*$   
incurred in working with the $c^{y^*}$-cost is $O(\OPT)$; this yields the 
$\bigl(O(1),O(1)\bigr)$-approximation.  

We extend the above idea to an arbitrary cost-minimization problem with packing
side-constraints as follows. 
Let $x^*$ be an optimal solution to the LP-relaxation, and $\Pc$ be the polytope obtained
by dropping the packing constraints.    
We observe that the same Lagrangian-duality based perturbation ensures that all points on
the minimal face of $\Pc$ containing $x^*$ have the same perturbed cost. 
Therefore, if we have an algorithm for the unweighted problem that rounds $x^*$ to a point
$\hx$ on this minimal face, then we again obtain that $(c^{y^*})^\intercal \hx=(c^{y^*})^\intercal x^*$,
which 
then leads to an $\bigl(O(1),O(1)\bigr)$-approximation (as in the case of \mcst). 

\paragraph{Related work.}
Whereas node-degree-bounded spanning-tree problems 
have been widely studied, relatively few results are known for 
spanning-tree problems with general degree constraints for a family $\Sc$ of node-sets.
With the exception of the result of~\cite{olver} for unweighted \cst, 
these other results~\cite{BansalKN09,AsadpourGMOS10,ChekuriVZ10,BansalKKNP13} all
yield a tree of cost at most the optimum with an $\w(1)$ additive- or multiplicative-
factor violation of the degree bounds. 
Both~\cite{BansalKN09} and~\cite{BansalKKNP13} obtain additive
factors via iterative rounding and relaxation. The factor in~\cite{BansalKN09} is $(r-1)$
for an arbitrary $\Sc$, 
where $r$ is the maximum number of degree constraints involving an edge (which could be
$|V|$ even when $\Sc$ is a chain), while~\cite{BansalKKNP13} yields an $O(\log|V|)$ factor
when $\Sc$ is a laminar family (the factor does not improve when $\Sc$ is a chain). The
dependent-rounding techniques in~\cite{AsadpourGMOS10,ChekuriVZ10} yield a tree $T$
satisfying $|\dt_T(S)|\leq\min\bigl\{O\bigl(\frac{\log|\Sc|}{\log\log|\Sc|}\bigr)b_S,
(1+\e)b_S+O\bigl(\frac{\log|\Sc|}{\e}\bigr)\bigr\}$ 
{for all $S\in\Sc$, for any family $\Sc$.} 

For \mbdst, 
Goemans~\cite{Goemans06}
obtained the first $\bigl(O(1),O(1)\bigr)$-approximation; his result yields a tree of
cost at most the optimum and at most $+2$ violation of the degree bounds. This was subsequently
improved to an (optimal) additive $+1$ violation
by~\cite{SinghL07}. Zenklusen~\cite{Zenklusen12} considers an orthogonal generalization of
\mbdst, where there is a matroid-independence 
constraint on the edges incident to each node, and obtains a tree of cost at most the
optimum and ``additive'' $O(1)$ violation (defined appropriately) of the matroid
constraints. To our knowledge, this is the only prior work that obtains an
$O(1)$-approximation to both the cost and packing constraints for a constrained
spanning-tree problem where an edge participates in $\w(1)$ packing constraints (albeit
this problem is quite different from spanning tree with general degree constraints).  

Finally, we note that our Lagrangian-relaxation based technique is somewhat similar to
its use in~\cite{KonemannR02}. However, whereas~\cite{KonemannR02} uses this to reduce
uniform-degree \mbdst to the problem of finding an MST of minimum 
maximum degree, which is another {\em weighted} problem, we utilize Lagrangian relaxation
in a more refined fashion to reduce the weighted problem to its {\em unweighted}
counterpart. 

\newcommand{\lp}[1]{\ensuremath{(\text{\ref{lp}}_{{#1}})}\xspace}
\newcommand{\dlp}[1]{\ensuremath{(\text{D}_{{#1}})}\xspace}
\newcommand{\lpld}{\lp{\ld}}
\newcommand{\dlpld}{\dlp{\ld}}
\newcommand{\lpldy}{\lp{\ld,y}}
\newcommand{\dlpldy}{\dlp{\ld,y}}
\newcommand{\ldlp}[1]{\ensuremath{(\text{LD}_{{#1}})}\xspace}

\section{An LP-relaxation for \mcst and preliminaries} \label{sec:lp_relaxation}
We consider the following natural LP-relaxation for \mcst. Throughout, we use $e$ to index
the edges of the underlying graph $G=(V,E)$. For a set $S\sse V$, let $E(S)$ denote
$\{uv\in E: u,v\in S\}$,
and $\dt(S)$ denote the edges on the boundary of $S$. 
For a vector $z\in\R^E$ and an edge-set $F$, 
{we use $z(F)$ to denote $\sum_{e\in F}z_e$.}
\begin{alignat}{3}
\min & \quad & \sum_e c_ex_e & \tag{P} \label{lp} \\
\text{s.t.}  && x\bigl(E(S)\bigr) & \leq |S|-1 \qquad && 
\forall \es\neq S\subsetneq V \label{acyc} \\
&& x(E) & = |V|-1 \label{span} \\
&& x\bigl(\delta(S)\bigr) & \le b_S \qquad && \forall S \in \mathcal{S} \label{deg} \\
&& x & \geq 0. \label{noneg}
\end{alignat}
For any $x\in\R_+^E$, let $\supp(x):=\{e:x_e>0\}$ denote the support of $x$.
It is well known that the polytope, $\pst(G)$, defined by \eqref{acyc},
\eqref{span}, and \eqref{noneg} is the convex hull of spanning trees of $G$. 
We call points in $\pst(G)$ {\em fractional spanning trees}. We refer to
\eqref{acyc}, \eqref{span} as the {\em spanning-tree constraints}. 
We will also utilize \lpld, the modified version of \eqref{lp} where 
we replace \eqref{deg} with $x\bigl(\dt(S)\bigr)\leq\ld b_S$ for all $S\in\Sc$, where
$\ld\geq 1$. 
{Let $\OPT(\ld)$ denote the optimal value of \lpld, and let $\OPT:=\OPT(1)$.}

\paragraph{Preliminaries.}
A family $\Lc\sse 2^V$ of sets is a {\em laminar family} if for all $A, B\in\Lc$, we have 
$A\sse B$ or $B\sse A$ or $A\cap B=\es$. As is standard, we say that $A\in\Lc$ is a child
of $L\in\Lc$ if $L$ is the minimal set of $\Lc$ such that $A\subsetneq L$. 
For each $L\in\Lc$, let $G^\Lc_L=(V^\Lc_L,E^\Lc_L)$ be the graph obtained from
$\bigl(L,E(L)\bigr)$ by contracting the children of $L$ in $\Lc$; we drop the superscript
$\Lc$ when $\Lc$ is clear from the context.
 
Given $x\in\pst(G)$, define a {\em laminar decomposition} $\Lc$ of $x$ to 
be a (inclusion-wise) maximal laminar family of sets whose spanning-tree constraints are
tight at $x$, so $x\bigl(E(A)\bigr)=|A|-1$ for all $A\in\Lc$. 
Note that $V\in\Lc$ and $\{v\}\in\Lc$ for all $v\in V$.  
A laminar decomposition can be constructed in polytime (using network-flow techniques).  
For any $L\in\Lc$, let $x^\Lc_L$, or simply $x_L$ if $\Lc$ is clear from context, denote $x$
restricted to $E_L$. 
Observe that $x_L$ is a fractional spanning tree of $G_L$.

\section{An LP-rounding approximation algorithm} \label{rounding}

\subsection{An overview} \label{overview}
We first give a high-level overview. 
Clearly, if \eqref{lp} is infeasible, there is no spanning tree satisfying
the degree constraints, so in the sequel, we assume that \eqref{lp} is feasible. 
We seek to obtain a spanning tree $T$ of cost $c(T)=O(\OPT)$ such that
$|\delta_T(S)|=O(b_S)$ for all $S \in \mathcal{S}$, 
{where $\dt_T(S)$ is the set of edges of $T$ crossing $S$.}

In order to explain the key ideas leading to our algorithm, we first briefly discuss
the approach of Olver and Zenklusen~\cite{olver} for unweighted \cst.  
Their approach {\em ignores} the edge costs $\{c_e\}$ and instead starts with a feasible
solution $x$ to \eqref{lp} that minimizes a suitable (linear) potential function. They use
this potential function to argue that if $\Lc$ is a laminar decomposition of $x$, then
$(x,\Lc)$ satisfies a key structural property called {\em rainbow freeness}. Exploiting
this, they give a rounding algorithm, hereby referred to as $\ozalg$, that for every
$L\in\Lc$, rounds $x_L$ to a spanning tree $T_L$ of $G_L$ such that
$|\dt_{T_L}(S)|\approx O\bigl(x_L(\dt(S))\bigr)$ for all $S\in\Sc$, 
so that concatenating the $T_L$s 
yields a spanning tree $T$ of $G$ satisfying $|\dt_T(S)|=O\bigl(x(\dt(S))\bigr)=O(b_S)$
for all $S\in\Sc$ (Theorem~\ref{ozthm}).  
However, as already noted in~\cite{olver}, a fundamental obstacle towards generalizing
their approach to handle the weighted version (i.e., \mcst) is that in order to
achieve rainbow freeness, which is crucial for their rounding algorithm, one needs to
{\em abandon the cost function $c$ and work with an alternate potential function}.  

We circumvent this difficulty as follows. 
First, we note that the algorithm in~\cite{olver} can be equivalently viewed as rounding
an {\em arbitrary} solution $x$ to \eqref{lp} as follows. 
Let $\Lc$ be a laminar decomposition of $x$.
Using the same potential-function idea, we can convert $x$ to another solution $x'$ to
\eqref{lp} that admits a laminar decomposition $\Lc'$ refining $\Lc$ such that $(x',\Lc')$
satisfies rainbow freeness (see Lemma~\ref{rfree}), and then round $x'$ using $\ozalg$.
Of course, the difficulty noted above remains, since the move to rainbow freeness
(which again ignores $c$ and uses a potential function) does not yield {\em any} bounds on
the cost $c^\intercal x'$ relative to $c^\intercal x$.  
We observe however that there is one simple property (*) under which one can guarantee that 
$c^\intercal x'=c^\intercal x$, namely, 
if for every $L\in\Lc$, all edges in $\supp(x)\cap E_L$ have the same cost. 
However, it is unclear how to utilize this observation since there is no reason to expect
our instance to have this rather special property: for instance, all edges of $G$ could
have very different costs!   

Now let $x^*$ be an optimal solution to \eqref{lp} (we will later modify this somewhat) 
and $\Lc$ be a laminar decomposition of $x^*$.
The crucial insight that allows us to leverage property (*), and a key notable aspect
of our algorithm and analysis, is that {\em one can leverage Lagrangian
duality to suitably perturb the edge costs so that the perturbed costs satisfy
property (*)}. 
More precisely, letting $y^*\in\R_+^\Sc$ denote the optimal values of the dual variables
corresponding to constraints \eqref{deg}, if we define the perturbed cost of edge $e$
to be $c^{y^*}_e:=c_e+\sum_{S\in\Sc:e\in\dt(S)}y^*_S$, then 
{\em the $c^{y^*}$-cost of all edges in $\supp(x^*)\cap E_L$ are indeed equal, for every 
$L\in\Lc$} (Lemma~\ref{equalize}). 
In essence, this holds because for any $e'\in\supp(x^*)$, by complementary slackness, we
have
$c_{e'}=\text{(dual contribution to $e'$ from \eqref{acyc},\eqref{span})}
-\sum_{S\in\Sc:e'\in\dt(S)}y^*_S$.
Since any two edges $e, f \in \supp(x^*)\cap E_L$ appear in the {\em same sets of $\Lc$},
one can argue that the dual contributions to $e$ and $f$ from \eqref{acyc}, \eqref{span}
are {\em equal}, and thus, $c^{y^*}_e=c^{y^*}_f$.
 
Now since $(x^*,\Lc^*)$ satisfies (*) with the perturbed costs $c^{y^*}$, we can convert
$(x^*,\Lc^*)$ to $(x',\Lc')$ satisfying rainbow freeness without altering the perturbed
cost, and then round $x'$ to a spanning tree $T$ using $\ozalg$. This immediately yields
$|\dt_T(S)|=O(b_S)$ for all $S\in\Sc$. To bound the cost, we argue that
$c(T)\leq c^{y^*}(T)=\sum_e c^{y^*}_ex^*_e=c^\intercal x^*+\sum_{S\in\Sc} b_Sy^*_S$
(Lemma~\ref{cbound}), where the last equality follows from complementary slackness.    
(Note that the perturbed costs are used only in the analysis.)   
However, $\sum_{S\in\Sc} b_Sy^*_S$ need not be bounded in terms of $c^\intercal x^*$. To fix this, we
modify our starting solution $x^*$: we solve \lpld (which recall is \eqref{lp} with
inflated degree bounds $\{\ld b_S\}$), where $\ld>1$, to obtain $x^*$, and use this $x^*$ in
our algorithm.  
Now, letting $y^*$ be the optimal dual values corresponding to the inflated degree
constraints, 
a simple duality argument shows that
$\sum_{S\in\Sc} b_Sy^*_S\leq\frac{\OPT(1)-\OPT(\ld)}{\ld-1}$ (Lemma~\ref{subgrad}), which 
yields $c(T)=O(\OPT)$ (see Theorem~\ref{mcstthm}).  

A noteworthy feature of our algorithm is the rather unconventional use of
Lagrangian relaxation, 
where we use duality to simplify the {\em cost structure} (as opposed to the
constraint-structure) of the underlying problem by equalizing edge costs in certain
subproblems. This turns out 
to be the crucial ingredient that allows us to utilize the algorithm of~\cite{olver} for
unweighted \cst{} {\em as a black box} without worrying about the difficulties posed by (the
move to) rainbow freeness. In fact, as we show in Sections~\ref{genredn} and~\ref{optcost}, this 
Lagrangian-relaxation idea is applicable more generally, and yields a novel reduction from
weighted problems to their unweighted counterparts. We believe that this reduction is of
independent interest and will find use in other settings as well; indeed, we demonstrate
another application of our ideas in Section~\ref{sec:kbmb}.

\subsection{Algorithm details and analysis}
To specify our algorithm formally, we first define the rainbow-freeness property and state
the main result of~\cite{olver} (which we utilize as a black box) precisely.  

For an edge $e$, define $\Sc_e:=\{S\in\Sc: e\in\dt(S)\}$. Note that $\Sc_e$ could be empty.
We say that two edges $e, f\in E$ form a \emph{rainbow} if $\Sc_{e}\sse\Sc_{f}$ or
$\Sc_{f}\sse\Sc_e$. 
(This definition is slightly different from the one in~\cite{olver}, in that we
allow $\Sc_e=\Sc_f$.) 
We say that $(x,\mathcal{L})$ is a \emph{rainbow-free decomposition} if $\Lc$ is a laminar 
decomposition of $x$ and for every set $L \in \mathcal{L}$, no two edges of 
$\supp(x)\cap E_L$ form a rainbow.  
(Recall that $G_L=(V_L,E_L)$ denotes the graph obtained from $\bigl(L,E(L)\bigr)$ by
contracting the children of $L$.)  
The following lemma shows that one can convert an arbitrary decomposition $(x,\Lc)$ to a
rainbow-free one; we defer the proof to the Appendix. 
(As noted earlier, this lemma is used to equivalently view the algorithm in~\cite{olver}
as a rounding algorithm that rounds an arbitrary solution $x$ to \eqref{lp}.) 

\begin{lemma} \label{lem:achieving_rainbow_freeness} \label{rfree}
Let $x\in\pst(G)$ and $\Lc$ be a laminar decomposition of $x$. 
We can compute in polytime a fractional spanning tree $x'\in\pst(G)$ and a
rainbow-free decomposition $(x',\Lc')$ such that:
(i) $\supp(x') \subseteq \supp(x)$;
(ii) $\mathcal{L} \subseteq \mathcal{L}'$; and
(iii) $x'(\delta(S)) \le x(\delta(S))$ for all $S \in \mathcal{S}$.
\end{lemma}

\begin{theorem}[\cite{olver}] \label{key1} \label{ozthm}
There is a polytime algorithm, $\ozalg$, that 
given a fractional spanning tree $x' \in \pst(G)$ and a rainbow-free decomposition 
$(x',\mathcal{L}')$, returns a 
spanning tree $T_L\sse\supp(x')$ of $G_L$ for every $L\in\Lc'$ such that  
the concatenation $T$ of the $T_L$s is a spanning tree of $G$ satisfying
{$|\delta_T(S)| \le 9 x'\bigl(\delta(S)\bigr)$ for all $S \in \mathcal{S}$.}
\end{theorem}

We can now describe our algorithm quite compactly. Let $\ld>1$ be a parameter.
\begin{enumerate}[nosep]
\item Compute an optimal solution $x^*$ to \lpld, a laminar decomposition $\Lc$ of $x^*$. 

\item Apply Lemma~\ref{rfree} to $(x^*,\Lc)$ to obtain a rainbow-free decomposition
$(x',\Lc')$. 

\item Apply $\ozalg$ to the input $(x', \mathcal{L}')$
to obtain spanning trees $T^{\Lc'}_L$ of $G^{\Lc'}_L$ for every $L\in\Lc'$. Return the
concatenation $T$ of all the $T^{\Lc'}_L$s.
\end{enumerate}

\paragraph{Analysis.} 

We show that the above algorithm satisfies the following guarantee.

\begin{theorem} \label{theo:main} \label{mcstthm}
The above algorithm run with parameter $\lambda > 1$ 
returns a spanning tree $T$ satisfying $c(T)\leq\frac{\ld}{\ld-1}\cdot\OPT$ and
$|\dt_T(S)|\leq 9\ld b_S$ for all $S\in\Sc$.
\end{theorem}

The bound on $|\dt_T(S)|$ follows immediately from Lemma~\ref{rfree} and
Theorem~\ref{ozthm} since $x^*$, and hence $x'$ obtained in step 2, is a feasible solution
to \lpld. So we focus on bounding $c(T)$. This will follow from three things. First, we
define the perturbed $c^{y^*}$-cost precisely. Next, Lemma~\ref{equalize} proves the key
result that for every $L\in\Lc$, all edges in $\supp(x^*)\cap E_L$ have the same perturbed
cost. Using this it is easy to show that 
$c(T)\leq c^{y^*}(T)=\sum_e c^{y^*}_ex^*_e=\OPT(\ld)+\ld\sum_{S\in\Sc} b_Sy^*_S$
(Lemma~\ref{cbound}). Finally, we show that   
$\sum_{S\in\Sc} b_Sy^*_S\leq\frac{\OPT-\OPT(\ld)}{\ld-1}$ (Lemma~\ref{subgrad}), which
yields the bound stated in Theorem~\ref{mcstthm}. 

To define the perturbed costs, we consider the Lagrangian dual of \lpld obtained by
dualizing the (inflated) degree constraints $x\bigl(\dt(S)\bigr)\leq\ld b_S$ for all
$S\in\Sc$: 
\begin{equation}
\max_{y \in \mathbb{R}_+^\mathcal{S}} \quad \biggl(g_\ld(y)\ \ :=\ \ 
\min_{x \in \pst(G)}\Bigl(\sum_e c_ex_e+\sum_{S\in\Sc}\bigl(x(\dt(S))-\ld b_S)y_S\Bigr)\biggr).
\tag{LD$_\ld$} \label{ldual} 
\end{equation}
For $y \in \mathbb{R}^\mathcal{S}$, 
let $\ldobj_{\lambda, y}(x) := 
\sum_{e}c_e x_e+\sum_{S \in \mathcal{S}} \bigl(x(\delta(S)) - \lambda b_S\bigr) y_S
= \sum_{e} c^y_e x_e - \ld\sum_{S \in \mathcal{S}} b_S y_S$
denote the objective function of the LP that defines $g_\ld(y)$, where 
$c^y_e := c_e + \sum_{S \in \mathcal{S} : e \in \delta(S)} y_S$. 

Let $y^*$ be an optimal solution to \eqref{ldual}. Our {\em perturbed costs} are $\{c^{y^*}_e\}$. We prove the following preliminary lemma, then proceed to show that the perturbed costs satisfy property (*).

\begin{lemma} \label{lpduality}
We have $g_{\lambda}(y^*) = \ldobj_{\lambda, y^*}(x^*)=\OPT(\ld)$.
\end{lemma}

\begin{proof}
For any $y\in\R_+^\Sc$, we have $g_{\lambda}(y)+\ld\sum_{S\in\Sc} b_Sy_S=$
\begin{eqnarray*}
\underbrace{\Bigl(\min \sum_e c^y_ex_e \quad \text{s.t.} \quad x\bigl(E(S)\bigr)\leq|S|-1
\ \ \forall \es\neq S\subsetneq V, \quad x(E)=|V|-1, \quad x\geq 0\Bigr)}_{\lpldy}\ = \\
\underbrace{\Bigl(\max -\sum_{\es\neq S\sse V}(|S|-1)\mu_S \quad \text{s.t.} \quad 
-\sum_{\substack{\es\neq S\sse V: \\ e\in E(S)}}\mu_S\leq c^y_e \ \ \forall e\in E, \quad
\mu_S\geq 0\ \ \forall \es\neq S\subsetneq V\Bigr)}_{\dlpldy}
\end{eqnarray*}
where the second equality follows since \dlpldy is the dual of \lpldy.
It follows that
\begin{alignat*}{3}
g_\ld(y^*)\ = \ \max_{y\in\R_+^\Sc}g_\ld(y)\ =\ 
\max & \ \ & -\sum_{\emptyset \neq S \subseteq V} (|S|-1)\mu_S & - \ld\sum_{S \in \mathcal{S}} b_S y_S 
\tag{D$_\ld$} \label{dlpld} \\
\text{s.t.} & \ \ & - \sum_{\substack{\emptyset \neq S \subseteq V: \\ e \in E(S)}} \mu_S
& -\sum_{\substack{S\in\Sc: \\ e\in\dt(S)}}y_S \leq c_e \quad && \forall e \in E \\
&& y\geq 0, \quad & \mu_S \ge 0 && \forall \emptyset \neq S \subsetneq V.
\end{alignat*}
Notice that \eqref{dlpld} is the dual of \lpld, hence, $g_\ld(y^*)=\OPT(\ld)$. Moreover,
it also follows that $\hy$ is an optimal solution to \eqref{ldual} iff 
there exists $\hmu=(\hmu_S)_{\es\neq S\subseteq V}$ such that 
$(\hmu,\hy)$ is an optimal solution to \eqref{dlpld}. 

So let $\mu^*$ be such that $(\mu^*,y^*)$ is an optimal solution to \eqref{dlpld}. It
follows that $x^*$ and $(\mu^*,y^*)$ satisfy complementary slackness. So we have that 
if $\mu^*_S>0$ then $x^*\bigl(E(S)\bigr)=|S|-1$, and
if $x^*_e>0$ then 
$-\sum_{\emptyset \neq S\subseteq V:e \in E(S)}\mu^*_S-\sum_{S\in\Sc:e\in\dt(S)}y^*_S=c_e$, or 
equivalently $c^{y^*}_e=-\sum_{\es\neq S\subseteq V:e\in E(S)}\mu^*_S$.
Therefore, 
\begin{align*}
\ldobj_{\ld, y^*}(x^*) & =\sum_e c^{y^*}_ex^*_e-\ld\sum_{S\in\Sc} b_Sy^*_S 
=\sum_e\Bigl(-\sum_{\es\neq S\subseteq V:e\in E(S)}\mu^*_S\Bigr)x^*_e-\ld\sum_{S\in\Sc} b_Sy^*_S \\
& =-\sum_{\es\neq S\sse V}\mu^*_Sx^*\bigl(E(S)\bigr)-\ld\sum_{S\in\Sc} b_Sy^*_S \\
& =-\sum_{\es\neq S\sse V}(|S|-1)\mu^*_S-\ld\sum_{S\in\Sc} b_Sy^*_S = g_\ld(y^*). \qedhere
\end{align*}
\end{proof}

\begin{lemma} \label{lem:same_modified_cost} \label{equalize}
For any $L \in \mathcal{L}$, all edges of $\supp(x^*)\cap E_L$ have the same $c^{y^*}$-cost.
\end{lemma}

\begin{proof}
Consider any two edges $e,f\in\supp(x^*)\cap E_L$.
Suppose for a contradiction that $c^{y*}_{e} < c^{y*}_{f}$. 
Obtain $\hat{x}$ from $x^*$ by increasing $x^*_{e}$ by $\epsilon$  
and decreasing $x^*_{f}$ by $\epsilon$ (so $\hx_{e'}=x^*_{e'}$ for all $e'\notin\{e,f\}$).
Using the same argument as in the proof of Lemma \ref{lem:achieving_rainbow_freeness}, one can show that $\hx\in\pst(G)$ for a sufficiently small $\e>0$.
Since $c^{y^*}_e<c^{y^*}_f$, we have 
$g_\ld(y^*)\leq \ldobj_{\ld, y^*}(\hx)<\ldobj_{\ld, y^*}(x^*)=g_\ld(y^*)$, where the last equality follows from Lemma \ref{lpduality}, and we obtain a
contradiction.   
\end{proof}

\begin{lemma} \label{cbound}
We have $c(T)\leq\sum_e c^{y^*}_ex^*_e=\sum_e c_ex^*_e+\ld\sum_{S\in\Sc}b_Sy^*_S$. 
\end{lemma}

\begin{proof}
Observe that $c(T)\leq c^{y^*}(T)$ since $c_e\leq c^{y^*}_e$ for all $e\in E$ as 
$y^*\geq 0$. We now bound $c^{y^*}(T)$.
To keep notation simple, we use $G_L=(V_L,E_L)$ and $x^*_L$ to denote $G^\Lc_L$ and
$(x^*)^\Lc_L$ (which recall is $x^*$ restricted to $E^\Lc_L$) respectively, and
$G'_L=(V'_L,E'_L)$ and $x^{*'}_L$ to denote $G^{\Lc'}_L$ and $(x^*)^{\Lc'}_L$ respectively.

We have $c^{y^*}(T)=\sum_{L\in\Lc}c^{y^*}(T\cap E_L)$ since the sets $\{E_L\}_{L\in\Lc}$
partition $E$.  
Fix $L\in\Lc$. 
Note that $x^*_L$ is a fractional spanning tree of $G_L=(V_L,E_L)$ since for any 
$\es\neq Q\sse V_L$, if $R$ is the subset of $V$ corresponding to $Q$, and $A_1,\ldots,A_k$
are the children of $L$ whose corresponding contracted nodes lie in $Q$, we have 
$x^*_L\bigl(E_L(Q)\bigr)=x^*\bigl(E(R)\bigr)-\sum_{i=1}^k x^*\bigl(E(A_i)\bigr)
\leq |R\sm(A_1\cup\ldots\cup A_k)|+k-1=|Q|-1$ with equality holding when $Q=V_L$. 
Note that $T\cap E_L$ is a spanning tree of $G_L$ since $T$ is obtained by concatenating  
spanning trees for the graphs $\{G'_{L'}\}_{L'\in\Lc'}$, and $\Lc'$ refines $\Lc$. 
Also, all edges of $\supp(x^*)\cap E_{L}$ have the same $c^{y^*}$-cost by
Lemma~\ref{equalize}. So we have $c^{y^*}(T\cap E_L)=\sum_{e\in E_L}c^{y^*}_ex^*_e$. 
It follows that 
\begin{align*}
c^{y^*}(T) & =\sum_e c^{y^*}_ex^*_e=\sum_e \Bigl(c_ex^*_e+\sum_{S\in\Sc:e\in\dt(S)}y^*_Sx^*_e\Bigr)
\notag \\
& =\sum_e c_ex^*_e+\sum_{S\in\Sc}y^*_Sx^*\bigl(\dt(S)\bigr)
=\sum_e c_ex^*_e+\ld\sum_{S\in\Sc}b_Sy^*_S. \qedhere  
\end{align*}
\end{proof}

\begin{lemma} \label{subgrad}
We have $\sum_{S\in\Sc} b_Sy^*_S\leq\frac{\OPT(1)-\OPT(\ld)}{\ld-1}$.
\end{lemma}

\begin{proof}
By Lemma~\ref{lpduality}, we have that 
\[
\OPT(\ld) = g_{\ld}(y^*) = \ldobj_{\ld, y^*}(x^*).
\]

Using Lemma~\ref{lpduality} again, now with $\lambda = 1$, and since $y^*$ is a feasible
solution to $\ldlp{1}$, we obtain that
$\OPT(1) = \max_{y \in \mathbb{R}_+^{\mathcal{S}}} g_1(y) \ge g_1(y^*)$.
Note that the objective functions of the LPs defining $g_1(y^*)$ and $g_{\ld}(y^*)$ differ
by a constant: 
$\ldobj_{1,y^*}(x) - \ldobj_{\ld,y^*}(x) = (\ld - 1) \sum_{S \in \mathcal{S}} b_Sy^*_S$
for all $x \in \pst(G)$. Since $x^*$ is an optimal solution to $\min_{x \in \pst(G)}
\ldobj_{\ld, y^*}(x)$, it is also an optimal solution to $\min_{x \in \pst(G)} \ldobj_{1,
  y^*}(x)$. It follows that 
\[
\OPT(1) \ge g_1(y^*) = \ldobj_{1, y^*}(x^*)  \ .
\]

Therefore,
$\OPT(1)-\OPT(\ld) \ge \ldobj_{1, y^*}(x^*) - \ldobj_{\ld, y^*}(x^*) 
= (\ld - 1) \sum_{S \in \mathcal{S}} b_Sy^*_S$.
\end{proof}

\begin{proofof}{Theorem~\ref{mcstthm}}
As noted earlier, the bounds on $\dt_T(S)$ follow immediately from Lemma~\ref{rfree} and
Theorem~\ref{ozthm}: for any $S\in\Sc$, we have
$|\dt_T(S)|\leq 9x'\bigl(\dt(S)\bigr)\leq 9x^*\bigl(\dt(S)\bigr)\leq 9\ld b_S$.
The bound on $c(T)$ follows from Lemmas~\ref{cbound} and~\ref{subgrad} since 
$\sum_ec_ex^*_e=\OPT(\ld)$.
\end{proofof}

\section{A reduction from weighted to unweighted problems} \label{genredn}
We now show that our ideas are applicable more generally, and yield bicriteria
approximation algorithms for any cost-minimization problem with packing side-constraints,  
provided we have a suitable approximation algorithm for the {\em unweighted}
counterpart. Thus, our technique yields a {\em reduction} from weighted to unweighted
problems, which we believe is of independent interest. 

To demonstrate this, we first isolate the key properties of the rounding algorithm
$\rndalg$ used above for unweighted \cst that enable us to use it as a black box to
obtain our result for \mcst; this will yield an alternate, illuminating explanation of
Theorem~\ref{mcstthm}. 
Note that $\rndalg$ is obtained by combining the procedure in Lemma~\ref{rfree} and
$\ozalg$ (Theorem~\ref{ozthm}). 
First, we of course utilize that $\rndalg$ is an approximation algorithm for unweighted 
\cst, so it returns a spanning tree $T$ such that $|\dt_T(S)|=O\bigl(x^*(\dt(S))\bigr)$ for 
all $S\in\Sc$. Second, we exploit the fact that $\rndalg$ returns a tree $T$ that 
{\em preserves tightness of all spanning-tree constraints that are tight at $x^*$}.
{\em This is the crucial property that allows us to bound $c(T)$}, since this implies (as
we explain below) that $c^{y^*}(T)=\sum_e c^{y^*}_ex^*_e$, which then yields the bound on
$c(T)$ as before.  
The equality follows because the set of optimal solutions to the LP
$\min_{x\in\pst(G)}\ldobj_{\ld, y^*}(x)$ is a face of $\pst(G)$; thus {\em all} points on the
{\em minimal} face of $\pst(G)$ containing $x^*$ are optimal solutions to this LP,  
and the stated property implies that the characteristic vector of $T$ lies on this minimal 
face.   
In other words, while $\ozalg$ proceeds by exploiting the notions of rainbow freeness and
laminar decomposition, these notions are not essential to obtaining our result; {\em any}
rounding algorithm for unweighted \cst satisfying the above two properties can be utilized
to obtain our guarantee for \mcst. 

We now formalize the above two properties for an arbitrary cost-minimization problem with
packing side-constraints, and prove that they suffice to yield a bicriteria guarantee.   
Consider the following abstract problem, where $\mathcal{P} \subseteq \mathbb{R}_+^n$ is a fixed polytope: given $c \in \mathbb{R}_+^n$, $A \in \mathbb{R}_+^{m \times n}$, and $b \in \mathbb{R}_+^m$, find
\[
\min \ \ c^\intercal x \quad \text{s.t.} \quad x \text{ is an extreme point of } \mathcal{P}, 
\quad Ax \le b. \tag{Q$^\mathcal{P}$} \label{absp}
\] 
Observe that we can cast \mcst 
as a special case of \eqref{absp}, by taking $\mathcal{P}=\pst(G)$ (whose extreme points
are spanning trees of $G$), 
$c$ to be the edge costs, and $Ax \le b$ to be the degree constraints. 
Moreover, by taking $\Pc$ to be the convex hull of a bounded set $\{x\in\Z_+^n: Cx\leq d\}$ 
we can use \eqref{absp} to encode a discrete-optimization problem.

\newcommand{\abslp}[1]{\ensuremath{(\text{R}^\Pc_{{#1}})}\xspace}
\newcommand{\abslpld}{\abslp{\ld}}

We say that $\rndalg$ is a \emph{face-preserving rounding algorithm} (\fpra) for
unweighted \eqref{absp} if given any point $x\in\mathcal{P}$,  
$\rndalg$ finds in polytime an extreme point $\hx$ of $\mathcal{P}$ such that:

\noindent
{(P1) \quad $\hx$ belongs to the minimal face of $\mathcal{P}$ that contains $x$.} 

\noindent
We say that $\rndalg$ is a {\em $\beta$-approximation \fpra} (where $\beta \ge 1$) if
we {\em also} have: 

\noindent
{(P2) \quad $A\hx \le \beta Ax$.}

\newcommand{\prop}[1]{(P{#1})\xspace}

Let \abslpld denote the LP $\min\bigl\{c^\intercal x: x\in\Pc, \ Ax\leq\ld b\bigr\}$;
note that \abslp{1} is the LP-relaxation of \eqref{absp}.
Let $\absopt(\ld)$ denote the optimal value of \abslpld, and let $\absopt:=\absopt(1)$.
We say that an algorithm is a $(\rho_1,\rho_2)$-approximation algorithm for \eqref{absp}
if it finds in polytime an extreme point $\hx$ of $\Pc$ such that
$c^\intercal \hx\leq\rho_1\absopt$ and $A\hx\leq\rho_2b$. 

\begin{theorem} \label{genrednthm}
Let $\rndalg$ be a $\beta$-approximation \fpra 
for unweighted \eqref{absp}. Then, given any $\lambda > 1$, one can obtain a 
$\bigl(\frac{\lambda}{\lambda-1}, \beta \lambda\bigr)$-approximation algorithm for
\eqref{absp} using a single call to $\rndalg$.
\end{theorem}

\begin{proofsketch}
We dovetail the algorithm for \mcst and its analysis. We simply compute an 
optimal solution $x^*$ to \abslpld and round it to an extreme point $\hx$ of $\Pc$ using
$\rndalg$. By property \prop{2}, it is clear that $A\hx\leq\beta(Ax^*)\leq\beta\ld b$. 

For $y\in\R_+^m$, define $c^y:=c+A^\intercal y$.
To bound the cost, as before, we consider the
Lagrangian dual of \abslpld obtained by dualizing the side-constraints $Ax\leq\ld b$.
\begin{equation*}
\max_{y \in \mathbb{R}_+^m}\Bigl(h_\ld(y)\ :=\ \min_{x\in\Pc}\absldobj_{\ld, y}(x)\Bigr), 
\quad \text{where}\ \ \absldobj_{\ld, y}(x) := 
(c^{y})^\intercal x-\ld y^\intercal b.
\end{equation*}
Let $y^*=\argmax_{y\in\R_+^m}h_\ld(y)$. 
We can mimic the proof of Lemma~\ref{lpduality} to show that $x^*$ is an optimal solution
to $\min_{x\in\Pc}\absldobj_{\ld, y^*}(x)$. The set of optimal solutions to this LP
is a face of $\Pc$. So all points on the minimal face of $\Pc$ containing $x^*$ are
optimal solutions to this LP. By property \prop{1}, $\hx$ belongs to this minimal face
and so is an optimal solution to this LP. 
So $(c^{y^*})^\intercal \hx=(c^{y^*})^\intercal x^*=c^\intercal x^*+(y^*)^\intercal Ax^*=\absopt(\ld)+\ld(y^*)^\intercal b$, where 
the last equality follows by complementary slackness. Also, by the same arguments as in 
Lemma~\ref{subgrad}, we have $(y^*)^\intercal b\leq\frac{\absopt(1)-\absopt(\ld)}{\ld-1}$. 
Since $c\leq c^{y^*}$, we have $c^\intercal \hx\leq(c^{y^*})^\intercal \hx\leq\frac{\ld}{\ld-1}\cdot\absopt$.
\end{proofsketch}

\section{Towards a \boldmath $\bigl(1,O(1)\bigr)$-approximation algorithm for \eqref{absp}} \label{optcost}
A natural question that emerges from Theorems~\ref{mcstthm} and~\ref{genrednthm} is
whether one can obtain a $\bigl(1,O(1)\bigr)$-approximation, i.e., obtain a solution of
{\em cost at most $\absopt$}  that violates the packing side-constraints by an (multiplicative)
$O(1)$-factor. 
Such results are known for degree-bounded spanning tree problems
with various kinds of degree constraints~\cite{Goemans06,SinghL07,BansalKN09,Zenklusen12},
so, in particular, it is natural to ask whether such a result also holds for \mcst.
(Note that for \mcst, the dependent-rounding techniques
in~\cite{AsadpourGMOS10,ChekuriVZ10} yield a tree $T$ with $c(T)\leq\OPT$ and 
$|\dt_T(S)|\leq\min\bigl\{O\bigl(\frac{\log|\Sc|}{\log\log|\Sc|}\bigr)b_S,
(1+\e)b_S+O\bigl(\frac{\log|\Sc|}{\e}\bigr)\bigr\}$ for all $S\in\Sc$.) 
We show that our approach is versatile enough to yield such a guarantee provided we
assume a stronger property from the rounding algorithm $\rndalg$ for unweighted
\eqref{absp}. 

Let $A_i$ denote the $i$-th row of $A$, for $i=1,\ldots,m$.
We say that $\rndalg$ is an {\em $(\al,\beta)$-approximation \fpra{}}
for unweighted \eqref{absp} 
if {\em in addition} to properties \prop{1} and \prop{2}, it satisfies:
\begin{enumerate}[(P3), nosep, topsep=0.25ex, labelwidth=\widthof{(P3)}, leftmargin=!]
\item\label{prop3} \quad it rounds a feasible solution $x$ to \abslp{\al} 
to an extreme point $\hx$ of $\Pc$ satisfying
$A_i^\intercal\hat{x}\ge\frac{A_i^\intercal x}{\alpha}$ for every $i$ such that
$A_i^\intercal x=\al b_i$.
\end{enumerate}

\vspace{0.75ex}
\noindent
(For \mcst, property \prop{3} requires that $|\dt_T(S)|\geq b_S$ for every set
$S\in\mathcal{S}$ whose degree constraint (in \lp{\al}) is tight at the fractional 
spanning tree $x$.)

\begin{theorem} \label{twosidethm}
Let $\rndalg$ be an $(\al,\beta)$-approximation \fpra 
for unweighted \eqref{absp}. Then, 
one can obtain a $(1,\alpha\beta)$-approximation algorithm for \eqref{absp} using a single 
call to \rndalg.
\end{theorem}

\begin{proof}
We show that applying the algorithm from Theorem~\ref{genrednthm} with $\lambda=\alpha$ yields the claimed
result. It is clear that the extreme point $\hx$ returned satisfies $A\hx\leq\al\beta b$. 
As in the proof of Theorem~\ref{genrednthm}, let $y^*$ be an optimal solution to
$\max_{y\in\R_+^m}h_\ld(y)$ (where $\ld=\al$). 
In Lemma~\ref{cbound} and the proof of Theorem~\ref{genrednthm}, we use the weak bound
$c^\intercal \hx\leq (c^{y^*})^\intercal \hx$. We tighten this to obtain the improved bound on $c^\intercal \hx$.
We have $c^\intercal \hx=(c^{y^*})^\intercal \hx-(y^*)^\intercal A\hx$, and 
\begin{equation*}
(y^*)^\intercal A\hx =\sum_{i:A_i^\intercal x^*=\ld b_i}y^*_i(A_i^\intercal \hx)\geq
\sum_{i:A_i^\intercal x^*=\ld b_i}\frac{y^*_iA_i^\intercal x^*}{\al} 
=\sum_{i:A_i^\intercal x^*=\ld b_i}y^*_ib_i=(y^*)^\intercal b.
\end{equation*}
The first and last equalities above follow because $y^*_i>0$ implies that 
$A_i^\intercal x^*=\ld b_i$. The inequality follows from property \prop{3}. 
Thus, following the rest of the arguments in the proof of Theorem~\ref{genrednthm}, we 
obtain that 
\begin{equation*}
c^\intercal \hx\leq(c^{y^*})^\intercal \hx-(y^*)^\intercal b=c^\intercal x^*+(\ld-1)(y^*)^\intercal b\leq\absopt(1). \qedhere
\end{equation*}
\end{proof}

\subsection{Obtaining an additive approximation for \eqref{absp} and cost at most
  $\absopt$ via an \fpra with two-sided additive guarantees} 

\newcommand{\abslpadditive}[1]{\ensuremath{(\text{W}^\Pc_{{#1}})}\xspace}
\newcommand{\lpdelta}{\abslpadditive{\Delta}}
\newcommand{\lpzero}{\abslpadditive{\vec{0}}}

We now present a variant of Theorem~\ref{twosidethm} that shows that we can achieve cost
at most $\absopt$ and additive approximation for the packing side constraints using an
\fpra with two-sided {\em additive} guarantees. 
We give an application of this result in 
Section~\ref{sec:kbmb}, where we utilize it to obtain improved guarantees for the
$k$-budgeted matroid basis problem.

\begin{theorem} \label{twosidethm_additive}
Let $\rndalg$ be an \fpra for unweighted \eqref{absp} that given $x\in\Pc$ returns an
extreme point $\hx$ of $\Pc$ such that $Ax - \Delta \le A\hx \le Ax + \Delta$, where
$\Delta \in \mathbb{R}_+^m$ may depend on $A$ and $c$ (but not on $b$). Using 
a single call to \rndalg, we can obtain an extreme point $\tx$ of $\mathcal{P}$ such that
{$c^\intercal \tx \le \absopt$ and $A\tx \le b + 2\Delta$.}
\end{theorem}

The above result is obtained via essentially the same arguments as those in
Theorems~\ref{genrednthm} and~\ref{twosidethm}. 
For a vector $\Delta \in \mathbb{R}_+^m$, let \lpdelta denote the LP
$\min\bigl\{c^\intercal x: x\in\Pc, \ Ax\leq b + \Delta \bigr\}$. Let $\vec{0}$ denote the
all-zeros vector, and note that \lpzero is the LP-relaxation of \eqref{absp}. Let
$\absopt(\Delta)$ denote the optimum value of \lpdelta, and let $\absopt :=
\absopt(\vec{0})$.  
The Lagrangian dual of \lpdelta obtained by dualizing the side-constraints 
$Ax\leq b +\Delta$ is 
\begin{equation*}
\max_{y \in \mathbb{R}_+^m}\Bigl(\varphi_\Delta(y)\ :=\ \min_{x\in\Pc} \Phi_{\Delta,y}(x)\Bigr), 
\tag{LD$_\Delta$}
\end{equation*}
where $\Phi_{\Delta, y}(x) := 
(c^{y})^\intercal x - y^\intercal (b + \Delta)$.
(Recall that $c^y := c + A^\intercal y$.) Let $x^*$ be an optimal solution to $\lpdelta$
and $y^*=\argmax_{y\in\R_+^m}\varphi_\Delta(y)$. We have the following variants of Lemmas
\ref{lpduality} and \ref{subgrad}. 

\begin{lemma} \label{lpduality_additive}
We have $\varphi_{\Delta}(y^*) = \Phi_{\Delta, y^*}(x^*)=\absopt(\Delta)$.
\end{lemma}

\begin{proof}
This follows by mimicking the arguments used in the proof of Lemma \ref{lpduality}.
\end{proof}

\begin{lemma} \label{subgrad_additive}
We have $(y^*)^\intercal \Delta \le \absopt(\vec{0}) - \absopt(\Delta)$.
\end{lemma}

\begin{proof}
We mimic the proof of Lemma \ref{subgrad}. By Lemma~\ref{lpduality_additive}, we have that 
\[
\absopt(\Delta) = \varphi_{\Delta}(y^*) = \Phi_{\Delta, y^*}(x^*)
\]
and
$\absopt(\vec{0}) = \max_{y \in \mathbb{R}_+^{\mathcal{S}}} \varphi_{\vec{0}}(y) \ge \varphi_{\vec{0}}(y^*)$.
Note that $\Phi_{\vec{0},y^*}(x) - \Phi_{\Delta,y^*}(x) = (y^*)^\intercal \Delta$, which
is independent of $x$. So since $x^*$ is an optimal solution to 
$\min_{x \in \Pc}\Phi_{\Delta, y^*}(x)$, it is also an 
optimal solution to $\min_{x \in \Pc} \Phi_{\vec{0}, y^*}(x)$. It follows that 
\[
\absopt(\vec{0}) \ge \varphi_{\vec{0}}(y^*) = \Phi_{\vec{0}, y^*}(x^*).
\]
Hence, 
$\absopt(\vec{0})-\absopt(\Delta)\ge\Phi_{\vec{0}, y^*}(x^*)-\Phi_{\Delta,y^*}(x^*)=(y^*)^\intercal\Delta$.
\end{proof}

\begin{proofof}{Theorem~\ref{twosidethm_additive}}
The algorithm simply computes an optimal solution $x^*$ to \lpdelta, and rounds it to an 
extreme point $\tx$ of $\mathcal{P}$ using algorithm $\rndalg$.

It is clear that $A\tx \le Ax^* + \Delta \le (b + \Delta) + \Delta = b + 2\Delta$. Next we
argue that $c^\intercal \tx \le \absopt$.  We have $c^\intercal \tx=(c^{y^*})^\intercal
\tx-(y^*)^\intercal A\tx$, and  
\begin{align*}
(y^*)^\intercal A\tx & = \sum_{i:A_i^\intercal x^*= b_i + \Delta_i}y^*_i(A_i^\intercal \tx)\geq
\sum_{i:A_i^\intercal x^*=b_i + \Delta_i}y^*_i(A_i^\intercal x^* - \Delta_i) \\
& =\sum_{i:A_i^\intercal x^*= b_i + \Delta_i}y^*_ib_i=(y^*)^\intercal b.
\end{align*}
By Lemma \ref{lpduality_additive}, $x^*$ is an optimal solution
to $\min_{x\in\Pc}\varPsi_{\Delta, y^*}(x)$. 
So all points on the minimal face of $\Pc$ containing $x^*$ are
optimal solutions to this LP. In particular, since $\tx$ belongs to this minimal face (by
property \prop{1}), $\tx$  is an optimal solution to this LP.
This observation, along with the inequality above, yields $c^\intercal \tx\leq (c^{y^*})^\intercal x^*-(y^*)^\intercal b = \absopt(\Delta) + (y^*)^\intercal \Delta$. Finally, using Lemma \ref{subgrad_additive} yields $c^\intercal \tx \le \absopt(\vec{0})$ as required.
\end{proofof}

\subsection{Application to $k$-budgeted matroid basis}
\label{sec:kbmb}

Here, we seek to find a basis $S$ of a matroid $M = (U, \mathcal{I})$ satisfying $k$
budget constraints $\{d_i(S) \le B_i\}_{1 \le i \le k}$,  
where $d_i(S) := \sum_{e \in S} d_i(e)$. 
Note that this can be cast a special case of \eqref{absp}, where $\mathcal{P} = \mathcal{P}(M)$ is the
basis polytope of $M$, the objective function encodes a chosen budget constraint (say
the $k$-th budget constraint), and $Ax \le b$ encodes the remaining budget constraints. 
We show that our techniques, combined with a recent randomized algorithm
of~\cite{BansalN16}, yields a (randomized) algorithm that, for any $\epsilon > 0$,
returns in $n^{O(k^{1.5}/\epsilon)}$ time a basis that (exactly) satisfies the chosen 
budget constraint, and violates the other budget constraints by (at most) a 
$(1 +\epsilon)$-factor, where $n := |U|$ is the size of the ground-set of $M$.  
This {\em matches} the current-best approximation guarantee of~\cite{Grandoni2014} (who
give a deterministic algorithm) {\em and} the current-best running time
of~\cite{BansalN16}. 

\begin{theorem}[\cite{BansalN16}] \label{bnthm}
For some constant $\nu > 0$, there exists a randomized \fpra, \bnalg, for unweighted $\QPM$  that rounds any 
$x\in \PM$ to an extreme point $\hx$ of $\PM$ such that
$Ax- \nu \sqrt{k} \Delta \le A\hx \le Ax + \nu \sqrt{k} \Delta$,
where $\Delta = (\max_{1\le j \le n} a_{ij})_{1 \le i \le k - 1} = (\max_{e \in U} d_{i}(e))_{1 \le i \le k - 1}$.
\end{theorem}

\begin{lemma}\label{bmbthm_lemma}
There exists a polytime randomized algorithm that finds a basis $S$ of $M$ such that $d_k(S)\leq B_k$, and $d_i(S)\leq B_i+ 2 \nu \sqrt{k} \max_{e \in U} d_{i}(e)$ for all
$1\leq i\leq k - 1$, or determines that the instance is infeasible.
\end{lemma}

\begin{proof}
As explained above, we cast the problem as a special case of \eqref{absp} by using the
$k$-th budget constraint as the objective function, and the remaining budget constraints
as packing side-constraints. If the LP-relaxation of \eqref{absp} is infeasible, then the
budgeted-matroid-basis instance is infeasible. Otherwise, the above guarantee follows by applying
Theorem~\ref{twosidethm_additive} with the algorithm \rndalg=\bnalg. 
\end{proof}

 Using ideas from \cite{BansalN16}, we combine the algorithm from Lemma \ref{bmbthm_lemma} with a partial enumeration step as follows. We say an element $e \in U$ is \emph{heavy} if the inequality $d_i(e) > \frac{\epsilon}{2\nu \sqrt{k}} B_i$ holds for at least one index $i \in \{1, \dots, k\}$. Let $H$ denote the set of all heavy elements. We state our algorithm below. Let $\epsilon > 0$ be a parameter.

\begin{enumerate}[nosep, topsep=0.5ex]
\item For every set $\widetilde{H} \subseteq H$ of size $|\widetilde{H}| \le \frac{2 \nu
  k^{1.5}}{\epsilon}$, we do the following.
\begin{enumerate}[nosep]
\item Let $M'$ be the matroid obtained from $M$ by contracting the elements of $\widetilde{H}$ and deleting the elements of $H \setminus \widetilde{H}$. \label{alg:a}
\item Compute residual budgets $B'_i := B_i - d_i(\widetilde{H})$, for $i \in \{1, \dots, k\}$. \label{alg:b}
\item Run the algorithm from Lemma \ref{bmbthm_lemma} on matroid $M'$ with budgets $\{B'_i\}_{1 \le i \le k}$.
\item If the algorithm succeeds (that is, if the LP that it attempts to solve is feasible), then let $T$ be the set of elements returned, and let $S := \widetilde{H} \cup T$. If $S$ is a basis of $M$, $d_k(S) \le B_k$, and $d_i(S) \le (1 + \epsilon) B_i$ for all $1 \le i \le k - 1$, then return $S$. \label{alg:d}
\end{enumerate}
\item If step 1 does not return any set $S$, then return that the instance is infeasible.
\end{enumerate}

\begin{theorem} 
\label{bmbthm}
The algorithm above, run with parameter $\epsilon > 0$, finds in
$n^{O(k^{1.5}/\epsilon)}$ time a basis $S$ of $M$ such that $d_k(S) \le B_k$ and $d_i(S)
\le (1 + \epsilon) B_i$ for all $1 \le i \le k - 1$, or determines that the instance is
infeasible.  
\end{theorem}

\begin{proof}
Note that the number of iterations is at most $n^{\frac{2 \nu k^{1.5}}{\epsilon}} =
n^{O(k^{1.5}/\epsilon)}$. Since steps 1(a)--1(d) run in $\poly(n)$ time,
the overall running time is $n^{O(k^{1.5}/\epsilon)}$ as claimed. 

If the instance is infeasible, then any outcome of the algorithm (infeasible, or a basis
$S$) is correct. (Note that due to the verification done at the end of step 1(d), any set
$S$ returned must have the required properties.) 
So assume that the instance is feasible, and let $S^*$ be a basis of $M$ that exactly
satisfies all the budget constraints. We argue that in this case the algorithm does indeed
return a basis with the desired properties. Let $H^* := S^* \cap H$ be the 
set of heavy elements that $S^*$ contains. Note that since a heavy element uses up at least
one budget to an extent greater than $\frac{\epsilon}{2\nu \sqrt{k}}$, and since $S^*$
satisfies all the $k$ budget constraints, we must have $|H^*| \le
\frac{k}{\frac{\epsilon}{2\nu\sqrt{k}}} = \frac{2\nu k^{1.5}}{\epsilon}$. Note that at the
iteration corresponding to $\widetilde{H} = H^*$ (if the algorithm reaches it), the set
$S^* \setminus H^*$ is feasible for the residual problem (with a matroid $M'$ and residual
budgets $\{B'_i\}$ defined in steps 1(a) and 1(b)). Further, note that this
set also certifies that the resulting set $S$ satisfies $d_k(S) = d_k(H^*) + d_k(T) \le
d_k(H^*) + d_k(S^* \setminus H^*) = d_k(S^*) \le B_k$. Finally, for every $i \in \{1,
\dots, k - 1\}$, we have 
\begin{align*}
d_i(S) & = d_i(H^*) + d_i(T) \le d_i(H^*) + B'_i + 2 \nu \sqrt{k} \max_{e \in U \setminus H} d_i(e) \\
& \le B_i + 2\nu \sqrt{k} \frac{\epsilon}{2\nu\sqrt{k}} B_i = (1 + \epsilon) B_i,
\end{align*}
and so the set $S$ will pass the verification done at step 1(d) and will be returned by
the algorithm. 
\end{proof}

\bibliographystyle{plain}

\appendix

\section{Proof of Lemma~\ref{rfree}}
This follows from essentially the same potential-function argument as used in~\cite{olver}
to obtain a rainbow-free solution.
Sort the edges of $\supp(x)$ in increasing order of $|\Sc_e|$ breaking ties arbitrarily. 
Let $e_1, e_2, \ldots, e_{k}$ denote this ordering. Let $w \in \mathbb{R}^E$ be any weight
function such that $w_{e_1} < w_{e_2} < \dots < w_{e_{k}}$ (e.g., $w_{e_i} = i$ for all $i$). 
Let $x'$ be an optimal solution to the following LP. (Note that the LP has variables $\{z_e\}_{e \in E}$, and that the $\{x_e\}_{e \in E}$ values are fixed.)
\begin{alignat*}{3}
\min & \quad & \sum_e w_e & z_e \tag{P'} \label{potlp} \\
\text{s.t.} & \quad & z\in\pst&(G), \quad z_e =0 \qquad  && \forall e\notin\supp(x) \\
&& z\bigl(\delta(S)\bigr) & \le x\bigl(\dt(S)\bigr) \qquad && \forall S \in \mathcal{S} \\
&& z\bigl(E(L)\bigr) & = |L|-1 \qquad && \forall L\in\Lc.
\end{alignat*}
Properties (i) and (iii) hold by construction. Since we force the spanning-tree
constraints corresponding to sets in $\mathcal{L}$ to be tight, we can start with $\Lc$ and
extend it to obtain a laminar decomposition $\Lc'$ of $x'$ that refines $\Lc$, so (ii) holds.  

It remains to show that $(x',\Lc')$ is a rainbow-free decomposition. Consider any set $L\in\Lc'$ and any
two edges $e,f\in \supp(x')\cap E^{\mathcal{L}'}_L$, and suppose that $e, f$ form a rainbow. Let
$w_e<w_f$, so we must have $\Sc_e\sse\Sc_f$. 
Now perturb $x'$ by adding $\e$ to $x'_e$ (the argument below will show that $x'_e<1$) and
subtracting $\e$ from $x'_f$, where $\e>0$ is chosen to be suitably small; 
let $x''$ be this perturbed vector. Clearly,
$w^Tx''<w^Tx'$, so if we show that $x''$ is feasible to \eqref{potlp}, then we obtain a
contradiction. Clearly, $\supp(x'')\sse\supp(x)$. Since $\Sc_e\sse\Sc_f$ it follows that   
$x''\bigl(\dt(S)\bigr)\leq x\bigl(\dt(S)\bigr)$ for all $S\in\Sc$. Also,
$x''\bigl(E(L)\bigr)=x'\bigl(E(L)\bigr)=|L|-1$ for all $L\in\Lc$. 

Finally, we show that $x''\in\pst(G)$ for a sufficiently small $\e>0$. (Hence, $x'_e<x''_e\leq 1$.) 
For $A\sse V$ such that
$x'\bigl(E(A)\bigr)<|A|-1$, we obtain $x''\bigl(E(A)\bigr)\leq|A|-1$ by taking $\e>0$
suitably small; for $A$ with $x'\bigl(E(A)\bigr)=|A|-1$, we obtain
$x''\bigl(E(A)\bigr)=|A|-1$ since the spanning-tree constraints for all $L\in\Lc'$ are
tight at ($x'$ and) $x''$ and these span all other tight spanning-tree constraints.
\qed

\end{document}